\documentclass[5p,preprint]{elsarticle}

\usepackage[utf8]{inputenc}
\usepackage[T1]{fontenc}
\usepackage{amsmath,amssymb,amsthm}
\allowdisplaybreaks 
\makeatletter
\let\cl@part\relax
\makeatother

\usepackage{wrapfig}

\usepackage[colorlinks=true,linkcolor=blue,citecolor=blue]{hyperref}
\usepackage[nameinlink,capitalize,noabbrev]{cleveref}
\usepackage{enumitem}

\usepackage{mathtools}

\newtheorem{thm}{Theorem}
\newtheorem{lem}[thm]{Lemma}

\newtheorem{prop}[thm]{Proposition}
\newtheorem{defn}[thm]{Definition}
\newtheorem{rem}[thm]{Remark}
\newtheorem{example}[thm]{Example}
\newtheorem{assumption}[thm]{Assumption}

\Crefname{thm}{Theorem}{Theorems}
\Crefname{lem}{Lemma}{Lemmata}
\Crefname{prop}{Proposition}{Propositions}
\Crefname{assumption}{Assumption}{Assumptions}
\crefname{example}{Example}{Examples}
\Crefname{defn}{Definition}{Definitions}
\Crefname{cor}{Corollary}{Corollaries}
\Crefname{figure}{Fig.}{Figs.}

\newcommand{\R}{\mathbb{R}}
\newcommand{\N}{\mathbb{N}}

\newcommand\ddt{\frac{\text{d}}{\text{d} t}}
\newcommand{\rbl}{\left (}
\newcommand{\rbr}{\right )}

\newcommand{\al}{\left \langle}
\newcommand{\ar}{\right \rangle}

\newcommand{\nl}{\left\|}
\newcommand{\nr}{\right\|}
\newcommand{\cbl}{\left\lbrace }
\newcommand{\cbr}{\right\rbrace }
\newcommand{\setdef}[2]{\cbl\ #1\ \left|\ \vphantom{#1} #2\ \right.\cbr}
\newcommand{\Norm}[2][ ]{\nl #2 \nr_{#1}}
\newcommand{\SNorm}[1]{\Norm[\infty]{#1}}
\newcommand{\dd}[2][ ]{\tfrac{\text{\normalfont d}#1}{\text{\normalfont d}#2}}

\usepackage{xcolor}
\usepackage{comment}

\usepackage{tikz}
\usepackage{pgfmath}
\usepackage{pgfplots}
\pgfplotsset{compat=1.18}
\usetikzlibrary{intersections}  \usepgfplotslibrary{fillbetween}
\usetikzlibrary{patterns}

\begin{document}

\begin{frontmatter}

\title{A model-free approach to control barrier functions for higher-order systems\tnoteref{funding}}

\tnotetext[funding]{This work was funded by the Deutsche Forschungsgemeinschaft (DFG, German Research Foundation)~-- Project-IDs 471539468, 544702565.}
\author[Ilmenau]{Lukas Lanza} \ead{lukas.lanza@tu-ilmenau.de}    
\author[ICL]{Johannes K\"ohler} \ead{j.kohler@imperial.ac.uk}              
\author[Ilmenau,Paderborn]{Dario Dennst\"adt} \ead{dario.dennstaedt@tu-ilmenau.de} 
\author[Halle]{Thomas Berger}\ead{thomas.berger@mathematik.uni-halle.de}
\author[Ilmenau]{Karl Worthmann} \ead{karl.worthmann@tu-ilmenau.de}

\address[Ilmenau]{Optimization-based Control Group, TU Ilmenau, 98693 Ilmenau, Germany} 
\address[ICL]{Department of Mechanical Engineering, Imperial College London, London, UK}
\address[Halle]{Institut f\"ur Mathematik, Martin-Luther-Universit\"at Halle-Wittenberg, 06120~Halle (Saale), Germany} 
\address[Paderborn]{Institut f\"ur Mathematik, Universität Paderborn, Warburger Str. 100, 33098 Paderborn, Germany} 
        

\begin{abstract}    
    Control barrier functions (CBFs) are a widely applied modular tool to ensure safe operation of nonlinear dynamical control systems. However, for their construction accurate knowledge of the system dynamics is typically needed. This requirement was recently alleviated for relative-degree-one systems using techniques from prescribed performance control (PPC) or funnel control (FC). This article extends the model-free CBF design to nonlinear systems of arbitrary relative degree. Moreover, we show with a simple example that a straightforward extension of existing results for relative-degree-one systems fails. Instead, we utilize novel techniques from funnel control to characterize a subset of the controls satisfying a CBF condition without requiring a dynamic model or state measurement. Finally, we demonstrate the applicability of our results on a seven degrees of freedom robotic manipulator with relative degree two.
\end{abstract}

\end{frontmatter}
               
\section{Introduction}
Satisfaction of safety-critical constraints during runtime is a key requirement in systems and control, which becomes particularly challenging when dealing with  
nonlinear and uncertain systems.  
Recent years have seen the development of a diverse set of control techniques designed to address these challenges~\cite{wabersich2023data,brunke2022safe}.
These approaches can often be  
implemented as modular safety filters, only adjusting desired control strategies as much as needed to ensure safety.
Despite these advancements, the effectiveness of such approaches commonly depends on the availability of accurate models for the system dynamics. 

Model-free control strategies such as PID controllers~\cite{han2009pid} or model-free reinforcement learning~\cite{sutton1998reinforcement} avoid reliance on complex models and are often scalable. 

However, these approaches generally lack rigorous safety guarantees.  
This paper proposes a novel approach by integrating  
safety techniques, specifically control barrier functions (CBFs), with a model‑free nonlinear controller, namely funnel control. The resulting framework provides a model‑free, modular safety mechanism for a general class of nonlinear dynamical systems.

\textbf{Related work}.
CBFs, initially introduced in~\cite{wieland2007constructive}, have proven effective in maintaining safety-critical constraints, and their intuitive formulation and modularity contributed to their popularity 
in robotics~\cite{ames2019control,ames2016control}. CBF-based controllers compute the control input by solving a quadratic program (QP) enforcing a lower bound on the derivative of the barrier function, thereby ensuring that the safe set is positively invariant and  asymptotically stable. 
CBFs are classically restricted to relative‑degree‑one constraints, whereas high‑order CBFs (HOCBFs) handle higher relative degrees by enforcing derivative constraints up to the required order~\cite{ames2019control,xiao2021high}. 
However,  
CBF implementations require a model for nonlinear dynamics, access to the state (or its estimate), and exhibit a complexity that scales with system dimensionality.
Literature has sought to address some of these limitations by studying robustness to bounded disturbances~\cite{jankovic2018robust,gurriet2018towards}, sector-bounded nonlinearities~\cite{buch2021robust}, and state-estimation uncertainties~\cite{dean2021guaranteeing}. 
Further insights have emerged from the analysis of input-to-state safety properties~\cite{alan2023control,kolathaya2018input}, and several approaches aim to extract necessary model information from data~\cite{taylor2020learning,dhiman2021control} -- albeit at the cost of increased implementation complexity. 
To counteract this complexity, \cite{cohen2024safety} explores the use of reduced-order models. For special classes of reduced-order kinematic systems, \cite{molnar2021model} proposes a model-free CBF-based controller that requires only minimal information about the system dynamics. 

Funnel control (FC) and prescribed performance control (PPC) are two closely related nonlinear control design frameworks that enforce predefined bounds on the tracking error w.r.t.\ a reference signal.
FC imposes time-varying performance bounds using error-dependent gains, guaranteeing that the tracking error remains within a prespecified \textit{performance funnel}~\cite{IlchRyan02b,berger2018funnel}.
In PPC, algebraic transformations of the tracking error are used to enforce prescribed transient and steady-state specifications~\cite{bechlioulis2008robust,kostarigka2011adaptive,bechlioulis2014low}.
Unlike classical CBF-based techniques, both are model-free in the sense that, although the systems considered must satisfy structural properties (such as a well-defined relative degree, stability of the internal dynamics and a high-gain property), 
knowledge 
of the 
model parameters is 
not required. 
Consequently, the same feedback design can be applied to a whole system class. However, FC and PPC lack the modularity of CBFs, which can act as an add‑on filter that minimally adjusts a desired control input while preserving safety alongside the primary control objective. 

Recently, the interplay between reciprocal CBFs and PPC has been analyzed in~\cite{namerikawa2024equivalence}, demonstrating that the auxiliary error variables in PPC can define a reciprocal CBF, and, conversely, that the PPC feedback law can be recovered from a reciprocal CBF formulation. 
In our recent work~\cite{LanzKoehl_CBFandFunnel}, we followed a similar approach, but employed FC and zeroing CBFs instead of PPC and reciprocal ones, to obtain  
a model-free CBF. 
A related zeroing CBF-based funnel controller has also been developed in~\cite{verginis2022funnel}. 
Notably, the approaches in~\cite{namerikawa2024equivalence,verginis2022funnel,LanzKoehl_CBFandFunnel}
are all limited to relative-degree-one systems, which excludes most applications. 
In the present paper we extend these results to systems of arbitrary relative degree.

\textbf{Contribution}.
In this work, we extend the model-free CBF design from~\cite{LanzKoehl_CBFandFunnel}.
Specifically, in~\cite{LanzKoehl_CBFandFunnel} we characterized a subset of inputs that satisfy the CBF condition for systems of relative degree one.
In this paper, we generalize this result to systems of arbitrary order.
One might expect that this approach can be generalized to systems with a higher relative degree using the well established concepts of HOCBFs~\cite{xiao2021high}. 
However, as we show in~\Cref{Ex:Failure},  this is not the case. 
This motivates our alternative approach: Essentially, starting with the output tracking error, we define auxiliary variables by an iteration similar to HOCBFs~\cite{xiao2021high}
-- the main difference being that these variables are vector-valued and shift the higher-order Lie-derivatives used in the literature to the construction of the auxiliary variables. 
Then, we construct a CBF of order one using these auxiliary variables. 
This enables the definition of model-free CBF-based control laws for a large class of higher-order systems in \Cref{Sec:characterizing_a_subset_of_feasible_controls}. 
The presented approach allows for the satisfaction of output constraints in a model-free fashion. 
We validate the approach by numerical simulations  
involving a seven degree of freedom manipulator.

\textbf{Nomenclature.}
For~$N \in \N$, we set $[N] = \{1,\ldots,N\}$.
For $x,y \in \R^n$, we use $\langle x, y \rangle = x ^\top y$ and the Euclidean norm is denoted by $\|x\| = \sqrt{\langle x,x \rangle}$. 
For $V \subseteq \R^m$, we denote by $C^k(V;\R^n)$ the set of $k$-times continuously differentiable functions ${f: V \to \R^n}$, and $C(V;\R^n):=C^0(V;\R^n)$.
For~$f:\R^m \to \R^n$,~$g:\R^k \to \R^m$ and~$x\in\R^k$, we use $(f \circ g)(x) := f(g(x))$.
For an interval~$I \subseteq \R$, $L^\infty(I; \R^{n})$ is the Lebesgue space of measurable essentially bounded functions $f : I\to\R^n$ with norm $\|f \|_{\infty} = \text{esssup}_{t \in I} \|f(t)\|$; for~$k \in \N$, $W^{k,\infty}(I;\R^{n})$ is the Sobolev space of all functions $f:I\to\R^n$ with $k$-th weak derivative $f^{(k)}$ and ${f,f^{(1)},\ldots,f^{(k)}\in L^\infty(I; \R^{n})}$.
${\mathcal{K} := \{ \alpha \in C(\R_{\ge 0};\R) \, | \, \alpha(0) = 0, \alpha\ \text{strictly \ increasing}\}}$,
$\mathcal{K}_\infty^e := \{ \alpha \in C(\R;\R) \, | \, \alpha(0) = 0, \alpha\ \rm{strictly \ increasing},  \\ \lim_{a\to \pm \infty}\alpha(a) = \pm \infty \}$.
For $b: \R_{\ge 0} \times \R^n \to \R^p$ differentiable, $(L_f b)(\cdot)$ denotes the Lie derivative of~$b$ along a vector field~$f:\R^n\to\R^n$, i.e., $(L_f b)(t,\zeta) = {\nabla_\zeta b(t,\zeta)^\top f(\zeta) \in\R^p}$, and, for sufficiently smooth $b$ and $f$, $L_f^k b := L_f(L_f^{k-1}b)$ for $k\ge 1$ with $L_f^0 b := b$. For $g = [g_1,\ldots,g_m]:\R^n\to\R^{n\times m}$, we set $L_g b = [L_{g_1} b,\ldots, L_{g_m} b]$, and use $\partial_t b(t,\zeta) := \tfrac{\partial}{\partial t}b(t,\zeta)$.

\section{System class and control objective}
In this section, we introduce the class of dynamical systems under consideration and state the control objective.
We further discuss the concept of high-order control barrier functions in order to motivate our approach.

We consider multi-input multi-output nonlinear systems in input-output form
\begin{subequations} \label{eq:System_r_InputOoutputForm}
    \begin{align}
        \dot \xi_i(t) &= \xi_{i+1}(t), \quad i \in [r-1], \nonumber \\ 
        \dot \xi_r(t) &= f(\xi(t),\eta(t)) + g(\xi(t),\eta(t)) u(t), \label{eq:r_th_order_system}   \\
        \dot \eta(t) &= Q(\xi(t),\eta(t)), \label{eq:internal_dynamics} 
    \end{align}
\end{subequations}
where $r\in\N$ is the order of the system, $\xi(t) = \big(\xi_1(t)^\top,\ldots,\xi_r(t)^\top\big)^\top\in\R^{rm}$ are the measured components of the state, 
$\eta(t)\in\mathbb{R}^q$ is the internal state, and $y(t) = \xi_1(t) \in \R^m$ is the output at time~$t \ge 0$.
The functions $f:\R^{rm}\times\R^{q}\to\R^m$, $g:\R^{rm}\times\R^{q}\to\R^{m\times m}$, $Q:\R^{rm}\times\R^{q}\to\R^{q}$ are assumed to be locally Lip\-schitz continuous.
The system can be influenced via the input~$u(t) \in \R^m$ at time $t\ge 0$.
The nonlinear functions $f$, $g$, and $Q$ are not assumed to be known and are not available for controller design. We only assume the availability of the instantaneous measurements of the output and its derivatives up to order $r-1$, i.e., $\xi(t)$ is available to the controller at 
time~$t$. Note that the dimension~$m \in \N$ of the input and the output coincide and are known, while the dimension~$q\in\N$ of the internal state~$\eta$ is unknown.  
For a concise notation in the later analysis, we define
\begin{equation} \label{eq:FandG}
    \begin{aligned}
            F : \R^{rm} \times \R^{q} & \to \R^{{rm+q}}, \\
        (\xi, \eta)  &\mapsto (
            \xi_2^\top,  \ldots,  \xi_r^\top,  f(\xi,\eta)^\top, Q(\xi,\eta)
        )^\top, \\
        G : \R^{rm} \times \R^{q} & \to \R^{(rm+q) \times m}, \\
        (\xi, \eta)  &\mapsto (
            0,  \ldots,  0,  g(\xi,\eta)^\top, 0
        )^\top .  
    \end{aligned}
\end{equation}
Then, system~\eqref{eq:System_r_InputOoutputForm} is equivalent to
\begin{equation}
    \ddt \begin{pmatrix}
        \xi(t) \\ \eta(t)
    \end{pmatrix} = F(\xi(t),\eta(t)) + G(\xi(t),\eta(t)) u(t).
\end{equation}
\begin{rem}
In this paper, we consider systems in input-output normal form~\eqref{eq:System_r_InputOoutputForm}.
It is also possible to consider nonlinear dynamics in state-space representation
\begin{align*}
    \dot x(t) &= \tilde F(x(t)) + \tilde G(x(t))  u(t), \\
    y(t) &= \tilde H(x(t))
\end{align*}
with unknown functions~$\tilde F,\tilde G, \tilde H$ of appropriate dimensions.
If such a system satisfies the conditions proposed in~\cite[Cor.~5.6]{byrnes1991asymptotic}, then there exists a state-space transformation to input-output form~\eqref{eq:System_r_InputOoutputForm}.~\hfill$\diamond$
\end{rem}
In the upcoming analysis, we ask a system~\eqref{eq:System_r_InputOoutputForm} to satisfy the following two assumptions.
\begin{assumption} \label[assumption]{ass:ID_BIBS}
    The internal dynamics~\eqref{eq:internal_dynamics} are bounded-input bounded-state stable, i.e., for all~$c_0 > 0$, there exists~$\bar q > 0$ such that, for all $\xi \in L^\infty(\R_{\ge 0};\R^{rm})$ and all $\eta^0\in \R^q$, we have
    \begin{equation*}
        \| \eta^0\| + \|\xi\|_\infty \le c_0 \implies \| \eta(\cdot;0,\eta^0,\xi) \|_\infty \le \bar q,
    \end{equation*}
    where $\eta(\cdot; 0,\eta^0,\xi)$ denotes the unique global solution of~\eqref{eq:internal_dynamics} with $\eta(0) = \eta^0$.
\end{assumption}
\Cref{ass:ID_BIBS} ensures that the unmeasured internal state~$\eta$ evolves within a compact, yet unknown set for bounded~$\xi$.

The following structural property links the order~$r \in \N$ of system~\eqref{eq:System_r_InputOoutputForm} to the concept of relative degree, cf.~\cite{byrnes1991asymptotic}.
\begin{assumption} \label[assumption]{Ass:g_positive}
 The map~$g: \R^{rm} \times \R^{q} \to \R^{m \times m}$ in~\eqref{eq:r_th_order_system} is pointwise positive definite, i.e., $\langle z, g(\xi,\eta) z \rangle > 0$ for all~$(\xi,\eta) \in \R^{rm} \times \R^{q}$ and all~$z \in \R^{m} \setminus \{0\}$.   
\end{assumption}

We do not assume symmetry of~$g$.
However, for many mechanical systems, the input distribution~$g$ is related to an inertia matrix, which is positive definite and symmetric by construction, see, e.g.,~\cite{spong1994partial}.
Thus, many mechanical systems satisfy \Cref{Ass:g_positive}, see also the example in \Cref{Sec:numerical_example}.
Note that, instead of positive definiteness of~$g$, 
negative definiteness may be assumed instead, which will only change the sign in the feedback law derived later.

In the following, we call a (locally) absolutely continuous function $(\xi,\eta): [0, \omega) \to \R^{rm+q}$, $\omega\in (0,\infty]$, with $\xi(0) =\xi^0$, $\eta(0)=\eta^0$ a solution (in the sense of Carathéodory) to~\eqref{eq:System_r_InputOoutputForm}, if it satisfies~\eqref{eq:System_r_InputOoutputForm} for almost all~${t \in [0,\omega)}$. A solution~${(\xi,\eta)}$ is said to be maximal, if it has no right extension that is also a solution; a maximal solution is global, if $\omega=\infty$ holds.

\subsection{Control objective}

We aim to design a feedback control law based on Control Barrier Functions (CBFs) such that the output~$y$ of system~\eqref{eq:System_r_InputOoutputForm} follows a given reference trajectory~$y_{\rm ref}$, satisfying prescribed (potentially time varying) 
safety requirements encoded by the performance function~$\psi$. 
In particular, we aim to guarantee that the output~$y$ satisfies 
\begin{equation} \label{eq:ErrorGuarantee}
    \forall \, t \ge 0 : \ \|y(t) - y_{\mathrm{ref}}(t)\| \le \psi(t).
\end{equation}
This is equivalent to ensuring that, for all~$t \ge 0$, the partial state $\xi_1(t)$ evolves within the safe set
\begin{equation} \label{eq:SafeSet}
\!\! \!   \mathcal{C} :=  \left\{ (t,\xi) \in \R_{\ge 0} \!\times\! \R^{rm} \, | \, \psi(t) \!-\!  \|\xi_1 \!-\! y_{\mathrm{ref}}(t)\| \ge 0 \right\}. 
\end{equation}
The performance requirement~\eqref{eq:ErrorGuarantee} is illustrated in \Cref{Fig:funnel_and_error}. 
Note that the safe set is independent of the evolution of the internal state~$\eta$, i.e., only output (and derivative) information is used to define the safe set. 
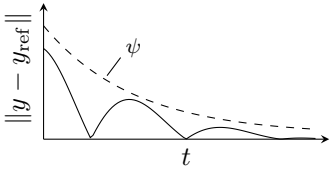
\begin{figure}[htbp]
    \centering
    \begin{tikzpicture}[>=stealth,scale=0.6]
    \draw[->] (0,0) --node[below]{$t$} (6.3,0);
    \draw[->] (0,0) -- node[rotate=90,left=-0.1em,above=0.1em]{$\|y-y_{\rm ref}\|$} (0,3) ;
    \draw[scale=2, dashed, domain=0:3, samples=100 , variable=\x, black] plot ({\x}, {1.2*exp(-\x)+0.05});
    \draw[scale=2, domain=0:3, samples=100 , variable=\x, black] plot ({\x}, {abs(cos(deg(0.5*\x))*cos(deg(3*\x))*exp(-0.7*\x)});
    \draw (1.4,1.4) -- node[above=0.7em,right=0.1em]{\footnotesize{$\psi$}} (1.7,1.8);
    \end{tikzpicture}
    \caption{Evolution of the tracking error~$y-y_{\rm ref}$ within the performance boundary~$\psi$.}
    \label{Fig:funnel_and_error}
\end{figure}

We consider the following conditions for the reference trajectory and the performance functions.
\begin{assumption}
\label[assumption]{Ass:reference_and_funnel}
    The reference trajectory~$y_{\rm ref}$ belongs to the set
    \begin{equation*}
    Y_{\rm ref} := 
    W^{r,\infty}(\R_{\ge 0}; \R^m). 
\end{equation*}
The performance function~$\psi$ is an element of the set
    \begin{equation*}
        \Psi := \left\{ \psi \in W^{1,\infty}(\R_{\ge 0};\R) \ \left|\ \inf_{s \ge 0} \psi(s) > 0 \right\}\right. .
    \end{equation*}
\end{assumption}

\begin{rem}
The considered problem is particularly meaningful in the context of motion planning problem in robotics. Specifically, in this case a collision-free reference is designed and the performance function $\psi$ can be chosen based on the distance of this reference to obstacles, i.e., satisfaction of~\eqref{eq:ErrorGuarantee} implies collision avoidance.
Algorithms to generate corresponding representations for complex cluttered environments are for example found in~\cite{wullt2025probabilistic}.
Lastly, the modular CBF-based control design will allow for additional excitation, e.g., to learn model parameters. 
These points will be illustrated with a rigid manipulator example in Section~\ref{Sec:numerical_example}.
~\hfill $\diamond$
\end{rem}

\subsection{High-order Control Barrier Functions}
To motivate our approach,  we briefly recall the concept of high-order CBFs (HOCBFs).  
Consider a system
\begin{equation} \label{eq:System_r}
    \begin{aligned}
         \dot x(t) &=  \tilde F(x(t)) +  \tilde G(x(t)) u(t), \quad x(0) = x^0 \!\in \! \R^n,
    \end{aligned}
\end{equation}
and a barrier function $B \in C^r( \R_{\ge 0} \times \R^n; \R)$. 
Suppose that the system~\eqref{eq:System_r} has 
relative degree~
$r\in\mathbb{N}$ w.r.t.\ 
$B(t,x)$, i.e., for all $(t,x)\in\mathbb{R}\times\mathbb{R}^n$:
\begin{align}\label{eq:rel_degree_condition}
    L_{\tilde G}^{} L_{\tilde F}^{{i-1}} B(t,x)=0~\forall\, i\in [r-1],~L_{\tilde G}^{} L_{\tilde F}^{r-1} B(t,x)\neq 0.
\end{align}
In virtue of~\cite{xiao2021high,xu2018constrained}, the following auxiliary functions are defined recursively for~$r \le n $ and 
$i \in [r-1]$ by
\begin{align}
\label{eq:FunctionsHOCBF}
\gamma_i(t,x)= \partial_t \gamma_{i-1}(t,x) + (L_{\tilde F} \gamma_{i-1})(t,x)+ \alpha_i(\gamma_{i-1}(t,x))
\end{align}
with~$\gamma_0(t,x) := B(t,x)$ for $t\ge 0$, $x\in\R^n$, and $\alpha_i \in \mathcal{K}_\infty^e$.
Note that $\gamma_i : \R_{\ge 0} \times \R^n \to \R$ are scalar-valued functions.
Given~$\gamma_i$, corresponding superlevel sets are defined by
\begin{equation} \label{eq:SuperlevelSets}
    \! C_i \!:=\! \{ (t, x) \in \R_{\ge 0} \times \R^n \, | \, \gamma_{i-1}(t,x) \ge 0 \}, \ i \in [r].
\end{equation}
Invoking the above functions and sets, the following definition is given in \cite[Def.~8]{xiao2021high}.

\begin{defn}[High-order CBF] \label[defn]{Def:HOCBF}
    Let~$C_i$ be defined as in~\eqref{eq:SuperlevelSets} and~$\gamma_i$ be given by~\eqref{eq:FunctionsHOCBF} 
    for $i \in [r-1]$.
    A sufficiently smooth function $B: \R_{\ge 0} \times \R^n \to \R$ is a candidate high-order CBF of order~$r$ for system~\eqref{eq:System_r}, if there exist sufficiently smooth~$\alpha_1,\ldots,\alpha_r \in \mathcal{K}_\infty^e$ such that
    \begin{equation} \label{eq:CriterionCBF}
    \begin{aligned}
        \sup_{{u\in\R^m}} \Big[ (L_{\tilde F}^r B)(t,x) + (L_{\tilde G}^{} L_{\tilde F}^{r-1}B)(t,x) u + \tfrac{\partial^r B(t,x)}{\partial t^r} \\
         + O(B(t,x)) + \alpha_r(\gamma_{r-1}(t,x)) \Big] \ge 0 
    \end{aligned}
    \end{equation}
    for all $(t,x) \in \cap_{i=1}^r C_i$, where
    \begin{equation*}
    \begin{aligned}
        O(B(t,x)) := \sum_{i=1}^{r-1} \Big[ L_{\tilde F}^i(\alpha_{r-i} \circ \gamma_{r-i-1})(t,x) \\ + \frac{\partial^i (\alpha_{r-i} \circ \gamma_{r-i-1})(t,x)}{\partial t^i} \Big].
    \end{aligned}
    \end{equation*}
\end{defn}
Given a HOCBF~$B$, the set of CBF-based controls is 
given by
\begin{equation} \label{eq:K_HOCBF}
    \begin{aligned}
     K_{\rm HO}(t,x) \!:=\! \Big\{ u \in {\R^m} \, | \, L_{\tilde F}^r B(t,x) \!+\! (L_{\tilde G^{}} L_{\tilde F}^{r-1} B)(t,x) u \\ + \tfrac{\partial^r B(t,x)}{\partial t^r} + O(B(t,x)) + \alpha_r(\gamma_{r-1}(t,x)) \ge 0 \Big\}.
    \end{aligned}
\end{equation}

The following example shows, why a straightforward extension of the ideas from our previous work~\cite{LanzKoehl_CBFandFunnel} to  higher-order systems is not expedient in general.
\begin{example} \label[example]{Ex:Failure}
Consider a controlled double integrator
\begin{align*}
    \dot z(t) &= \begin{pmatrix}
        z_2(t) \\ 0
    \end{pmatrix} + \begin{bmatrix}
        0 \\1
    \end{bmatrix} u(t), \quad z(0) = z^0 \in \R^2, 
\end{align*}
thus, $ F(z)=(z_2,0)^\top$, and $ G(z) = (0,1)^\top$. 
Following the results in~\cite{LanzKoehl_CBFandFunnel}, with the aim to keep the output $y(t) = z_1(t)$ within the interval $(-1,1)$, let $B(t,z) = \tfrac{1}{2}(1 - z_1^2)$.
Note that, in this particular case, $\partial_t B(t,z) = 0$. In order to verify~\eqref{eq:CriterionCBF}, let~$\alpha_i \in \mathcal{K}_\infty^e$, $i \in [2]$, such that $\alpha_1$ is continuously differentiable.  Following the construction in~\eqref{eq:FunctionsHOCBF}, we obtain $\gamma_0(t,z) = B(t,z)$ and 
\begin{align*}
    \gamma_1(t,z) &= (L_F \gamma_0)(t,z) + \alpha_1(\gamma_0(t,z))\\
    &=  -  z_1 z_2 + \alpha_1(\tfrac{1}{2}(1 - z_1^2)).
\end{align*}
Furthermore, we obtain for the quantities in~\eqref{eq:CriterionCBF}
\begin{align*}
O(B(t,z)) &= L_F (\alpha_1\circ \gamma_0)(t,z) + \partial_t(\alpha_1\circ \gamma_0)(t,z)\\
&=\alpha_1'(B(t,z)) \cdot [(L_F   +  \partial_t)B](t,z) \\
&= -\alpha_1'(B(t,z)) z_1 z_2,\\
    L_F^2B(t,z) &= -z_2^2, \\
    L_G L_FB(t,z) &= -  z_1.
\end{align*}
Thus, we calculate
\begin{align*}
    & L_{ F}^2 B(t,z) + (L_{ G} L_F B)(t,z) u + \partial_t^2 B(t,z) \\
    &\quad + O(B(t,z)) + \alpha_2(\gamma_1(t,z)) \\
    &= - z_2^2 -  z_1 u -  \alpha_1'(B(t,z))  z_1 z_2  \\
    &\quad + \alpha_2\big(-  z_1 z_2 + \alpha_1(\tfrac{1}{2}(1 - z_1^2))\big) \\
    & \stackrel{z_1=0}{=} - z_2^2 + \alpha_2\big(\alpha_1(\tfrac{1}{2})\big)
\end{align*}
which is negative for $z_1=0$ and $|z_2|>\sqrt{\alpha_2\big(\alpha_1(\tfrac{1}{2})\big)}$, independent of the choice of $u$. Indeed, all of those points are contained in $C_1\cap C_2$ given by
\begin{multline*}
      \big\{ (t,z)\in\R_{\ge 0}\times\R^2 \mid |z_1|<1,
     z_1 z_2 \le \alpha_1(\tfrac{1}{2}(1 - z_1^2)) \big\}.
\end{multline*}
Thus, condition~\eqref{eq:CriterionCBF} in \Cref{Def:HOCBF} is not satisfied for any choice of $\alpha_1,\alpha_2\in\mathcal{K}_\infty^e$ and hence~$B$ is not a high-order CBF.\\
We may further observe that, equipped with~$B(t,z(t))$, the system does not have a global relative degree two, as $L_{ G} L_{ F}B(t,z) = -  z_1$ is not invertible at $z_1=0$, which already violates one of the prerequisites for a HOCBF. 
Moreover, for any differentiable function $\tilde B:\R^2\to\R$ which satisfies~\eqref{eq:rel_degree_condition}, we have $\nabla_{z_2} \tilde B(z) = 0$, i.e., $\tilde B = \tilde B(z_1)$. Thus, with $\tilde B(z_1)\ge 0$ for all~$z_1\in(-1,1)$ and $\tilde B(z_1)\le 0$ otherwise, we find 
$L_{ G} L_{F} B(t,z) = \tilde B'(z_1)$ for $B(t,z) = \tilde B(z_1)$, which is zero for some ${z_1\in (-1,1)}$ by the intermediate value theorem. Therefore, the relative degree condition is violated for any possible choice of~$B$.~\hfill $\diamond$
\end{example}
The above example shows that there does not exist any barrier function candidate which is able to keep the output of the system within a desired margin and, at the same time, satisfies the definition of a HOCBF. We will resolve this drawback in the following with an alternative construction of CBFs for such cases.

\section{CBFs for higher-order systems using insights from funnel control} \label{Sec:HOCBFusingFunnelControl}

In this section, we introduce and characterize a set of CBF-based controls without knowledge about the system parameters~$f,g$, and $Q$ in~\eqref{eq:System_r_InputOoutputForm}.
\Cref{Ex:Failure} exhibits that the construction of the auxiliary variables~$\gamma_i$ in~\eqref{eq:FunctionsHOCBF} does not lead to a HOCBF according to \Cref{Def:HOCBF}. 
However, we may define vector-valued auxiliary variables~$\beta_i$ with a structure similar to~\eqref{eq:FunctionsHOCBF} and then, using these auxiliary variables, define a CBF for system~\eqref{eq:System_r_InputOoutputForm}, which necessitates a significantly different analysis.

\subsection{CBF candidate for higher-order systems}
\label{Sec:CBF_candidate_for_high_order_systems}
In this section, we use ideas from funnel control~\cite{BergIlch21} to construct auxiliary variables with a specific structure to define a CBF for higher-order systems.
The key insight is to shift the higher-order Lie-derivatives used in \Cref{Def:HOCBF} to the construction of auxiliary variables~$\beta_k$ and their respective domains.

Let $y_{\mathrm{ref}}$ and $\psi_k$, $k \in [r]$, satisfy \Cref{Ass:reference_and_funnel}. 
We assign each output derivative~$\xi_k$ to its own performance function~$\psi_k$.
Setting $\beta_0 \equiv 0$ and $\mathcal{D}_0 = \R_{\ge 0}\times\R^{rm}$, we recursively define auxiliary signals~$\beta_k: \mathcal{D}_{k} \to \R^m$, 
\begin{equation}\label{eq:AuxBeta}
    (t,\xi) \mapsto \frac{\xi_{k}-y_{\mathrm{ref}}^{(k-1)}(t)}{\psi_{k}(t)} + \frac{\beta_{k-1}(t,\xi)}{1-\|\beta_{k-1}(t,\xi)\|^2}
\end{equation}
on the domain $\mathcal{D}_{k} = \{ (t,\xi) {\in\mathcal{D}_{k-1}}\ | \ \|\beta_{k-1}(t,\xi)\| < 1 \}$. 
Since $\mathcal{D}_r \subset \mathcal{D}_{r-1} \subset \cdots \subset \mathcal{D}_1$ holds, the domains are nested; see~\Cref{fig:domain_auxvar} for a visualization of the domains.
\begin{figure}[ht]
    \centering
 \begin{tikzpicture}
\begin{axis}[
    scale=0.7,
    axis lines = middle,
    xlabel = $\xi_1$,
    ylabel = $\xi_2$,
    samples = 30,  
    ymax=5,
    ymin=-5,
    xmax=1.25,
    xmin=-1.25,
    xticklabel style={xshift=-2pt},
]

\addplot[name path=U, domain=-0.98:0.98, black, thick] 
    {1 - x/(1 - x^2)};
\addplot[name path=L, domain=-0.98:0.98, black, thick] 
    {-1 - x/(1 - x^2)};

\addplot[name path=ri, domain=-5:5, black, thick] ({1}, x);
\addplot[name path=le, domain=-5:5, black, thick] ({-1}, x);

\addplot [
    fill=white,
    draw=none
] fill between [of=U and L];

\addplot [
    pattern=north east lines,
    pattern color=black,
    draw=none
] fill between [of=U and L];

\end{axis}
\end{tikzpicture}
    \caption{Domains of the auxiliary variables~$\beta_1, \beta_2$ from~\eqref{eq:AuxBeta}
    for scalar~$\xi_1,\xi_2$, $y_{\mathrm{ref}} \equiv 0$, and~$\psi_1 = \psi_2 \equiv 1$. 
    The blank area~$-1 \le \xi_1 \le 1$ is a snapshot of~$\mathcal{D}_1$ for a fixed time~$t$, and the crosshatched area is a snapshot of~$\mathcal{D}_2$,  i.e., $| \xi_2 + \tfrac{\xi_1}{1 - \xi_1^2} | < 1$.%
    }
    \label{fig:domain_auxvar}
\end{figure}
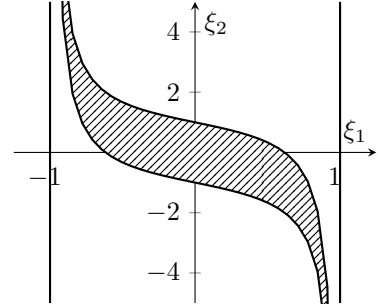

We define the set 
\begin{equation} \label{eq:control_domain}
    {\mathcal{D} := \{(t,\xi) \in \mathcal{D}_r \, | \, \|\beta_r(t,\xi) \| < 1 \}} \subset \mathcal{D}_r
\end{equation} and obtain the inclusion ${\mathcal{D} \subset \mathcal{C}}$ for the safe set~$\mathcal{C}$ from~\eqref{eq:SafeSet}.
Thus, ensuring that $(t,{y(t),\dot y(t),\ldots,y^{(r-1)}(t)}) \in \mathcal{D}$ for all~$t \ge 0$ implies safety, i.e., achievement of the control objective.

In \Cref{Def:HOCBF}, a CBF~$B(t,x)$ is written in terms of the entire state~$x$.
Since we consider higher-order systems in input-output form~\eqref{eq:System_r_InputOoutputForm}, i.e., we assume only {availability of} the output and its derivatives, we
propose the following candidate zeroing CBF
\begin{equation} \label{eq:CBF}
    \begin{aligned}
        b : \mathcal{D}_r &\to \R , \ 
        (t,\xi)  \mapsto \tfrac{1}{2} \left(1 - \| \beta_{r}(t,\xi)\|^2 \right).  
    \end{aligned}
\end{equation}
We will show that
it defines a CBF of order one with respect to the dynamics~\eqref{eq:System_r_InputOoutputForm}, i.e., {there exists $\alpha\in \mathcal{K}_\infty^e$ such that}
        \begin{equation} \label{eq:CBF-condition}
            \begin{aligned}
            \sup_{{u\in\R^m}} \Big[ \partial_t b(t,\xi) &+ (L_F b)(t,\xi,\eta) \\
            &+ (L_G b)(t,\xi,\eta)u + \alpha(b(t,\xi)) \Big] \ge 0
            \end{aligned}
        \end{equation}
for all~$(t,\xi,\eta) \in \mathcal{D}_r \times \R^q$, {where $F, G$ are defined in~\eqref{eq:FandG}.}
\begin{rem}
    We highlight that the function~$b$ in~\eqref{eq:CBF} is independent of the internal state~$\eta$.
    However, taking Lie-derivatives along~$F$ and~$G$ leads to expressions depending  
    on~$\eta$, e.g., 
    \begin{align*}
       (L_F b)(t,\xi,\eta) = [\nabla_\xi b(t,\xi)^\top, \underbrace{\nabla_\eta b(t,\xi)^\top}_{=0} ] F(\xi,\eta).
    \end{align*}
    To emphasize that~$b$ is independent of the internal state~$\eta$, we continue to write~$b(t,\xi)$, while taking the Lie-derivative requires to formally include~$\eta$.~\hfill $\diamond$
\end{rem}
Based on~\eqref{eq:CBF-condition}, we consider the set of CBF-based controls for a fixed~$\alpha \in \mathcal{K}_\infty^e$: 
\begin{align}
\label{eq:CBFbasedControls}
    K_{\rm CBF}(t,\xi,\eta,\alpha) &:= \Big\{ u \in \R^m \ | \ \partial_t b(t,\xi) + (L_F b)(t,\xi,\eta) \nonumber
    \\ &+ (L_G b)(t,\xi,\eta) u  \ge - \alpha(b(t,\xi)) \Big\},
\end{align}
which coincides with the set~$K_{\rm HO}(t,\xi,\eta)$ in~\eqref{eq:K_HOCBF} for~${r=1}$ in~\eqref{eq:System_r}, when $b=B$, $F=\tilde F$, $G=\tilde G$, and $\alpha=\alpha_r$.

We emphasize that, although the set~$K_{\rm CBF}$ formally depends on the unobserved internal state~$\eta$, we will construct a subset of feasible controls which is independent of the internal state.
To this end, we first show that the function defined in~\eqref{eq:CBF} satisfies the CBF condition~\eqref{eq:CBF-condition}.
\begin{prop}\label[prop]{Prop:bIsOutputCBF}
    Let \Cref{Ass:reference_and_funnel} hold for $y_{\mathrm{ref}} \in Y_{\rm ref}$ and $\psi_k$, $k \in [r]$. 
    Then, for system~\eqref{eq:System_r_InputOoutputForm} satisfying \Cref{ass:ID_BIBS,Ass:g_positive}, the function~$b$ in~\eqref{eq:CBF} satisfies the first-order CBF condition~\eqref{eq:CBF-condition}. 
\end{prop}
Before we prove \Cref{Prop:bIsOutputCBF}, we point out  
that system~\eqref{eq:System_r_InputOoutputForm} satisfying 
\Cref{Ass:g_positive} has relative degree~$r$ in the sense of \cite[Def.~3.1]{byrnes1991asymptotic}. 
The proposed auxiliary variable~$\beta_r$, on which the barrier function candidate~\eqref{eq:CBF} is based, has relative degree one w.r.t.\ the dynamics~\eqref{eq:System_r_InputOoutputForm}, which can be seen by calculating $(L_G \beta_r)(t,\xi,\eta) = g(\xi,\eta)/\psi_r(t)$.

\begin{proof}
    The function~$b$ is continuously differentiable on~$\mathcal{D}_r$.
    Fix $(t,\xi,\eta)\in\mathcal{D}_r\times\R^q$. Using~\eqref{eq:System_r_InputOoutputForm} and the functions~$F,G$ from~\eqref{eq:FandG}, we have that
    \begin{align*}
         &(L_F b)(t,\xi,\eta)= \left\langle \nabla_{\xi} b(t,\xi),  
         (\xi_2^\top \dots\ 
         \xi_r^\top\ f(\xi,\eta)^\top )^\top  
         \right\rangle \\
        & \qquad \qquad \qquad + \langle \nabla_{\eta} b(t,\xi), Q(\xi,\eta)\rangle\\
        &= \sum_{i=1}^{r-1} \langle \nabla_{\xi_i} b(t,\xi), \xi_{i+1}\rangle  + \langle \nabla_{\xi_r} b(t,\xi), f(\xi,\eta) \rangle\\
        &= -\sum_{i=1}^{r-1} \langle \big(\nabla_{\xi_i} \beta_r(t,\xi)\big) \beta_r(t,\xi), \xi_{i+1}\rangle\\
        &\quad - \langle \big(\nabla_{\xi_r} \beta_r(t,\xi)\big) \beta_r(t,\xi), f(\xi,\eta) \rangle\\
        &= -  \left\langle \beta_r(t,\xi), \big(\nabla_\xi \beta_r(t,\xi)\big)^\top \begin{pmatrix} \xi_2\\ \vdots\\ \xi_r\\ f(\xi,\eta)\end{pmatrix} + 0^\top \!\cdot\! Q(\xi,\eta) \right\rangle\\
        &= -  \langle \beta_r(t,\xi), (L_F \beta_r)(t,\xi,\eta)\rangle,
    \end{align*}
    where we used $\nabla_\eta \beta_r(t,\xi) = \nabla_\eta b(t,\xi) = 0 \in \R^{q \times m}$.
    For some~$u \in \R^m$, we have
    \begin{align*}
        \! (L_G b)(t,\xi,\eta) u &= \left\langle \nabla_{\xi} b(t,\xi), \!\begin{bmatrix} 0\\ g(\xi,\eta)\end{bmatrix} \! u\right\rangle 
        \!+\! \langle \nabla_{\eta} b(t,\xi), 0\rangle\\
        &= \langle \nabla_{\xi_r} b(t,\xi), g(\xi,\eta) u\rangle\\
        &= -\langle \big(\nabla_{\xi_r} \beta_r(t,\xi)\big) \beta_r(t,\xi), g(\xi,\eta) u\rangle \\
        &= -\langle \beta_r(t,\xi), g(\xi,\eta) u \rangle /\psi_r(t).
    \end{align*}
    Choosing the control 
    \begin{equation*}
        u = -\psi_r(t) g(\xi,\eta)^{-1} [(\partial_t + L_F + I)\beta_r](t,\xi,\eta) \in\R^m,
    \end{equation*}
    we calculate
        \begin{equation*}
        \begin{aligned}
            & \partial_t b(t,\xi) +  (L_Fb)(t,\xi,\eta) + (L_G b)(t,\xi,\eta) u \\
            &= 
            -\langle \beta_r(t,\xi), [(\partial_t + L_F)\beta_r](t,\xi,\eta) + {g(\xi,\eta) u /\psi_r(t)} \rangle \\
            & = - \langle \beta_r(t,\xi), [(\partial_t + L_F )\beta_r](t,\xi,\eta) \\
            & \qquad - [(\partial_t + L_F + I)\beta_r](t,\xi,\eta) \rangle \\
            & = \langle \beta_r(t,\xi), \beta_r(t,\xi) \rangle = \|\beta_r(t,\xi)\|^2 \\
            &> \|\beta_r(t,\xi)\|^2 - 1 = - 2 \alpha(b(t,\xi)), 
        \end{aligned}
    \end{equation*}
    where we use the function $\alpha:( s\mapsto 2 s) \in {\mathcal{K}_\infty^e}$. Hence, $b$ is a CBF candidate satisfying condition~\eqref{eq:CBF-condition}.
\end{proof}

\subsection{Characterizing a subset of feasible controls}
\label{Sec:characterizing_a_subset_of_feasible_controls}

In this section, we propose, based on the function~$b$ introduced in~\eqref{eq:CBF}, a method to characterize a model-free subset of feasible CBF-based controls.
To do so, we first present a fundamental property of the auxiliary variables~$\beta_k$ given in~\eqref{eq:AuxBeta} on the control domain~$\mathcal{D}$ given in~\eqref{eq:control_domain}. 
\begin{lem} \label[lemma]{Lem:bounds_on_derivatives}
    Consider~$b$ from~\eqref{eq:CBF}.
    For any compact subset~$ \bar{\mathcal{D}} \subset \mathcal{D}$ there exist constants~$B_i \ge 0$, $i \in [r]$, and~${T \ge 0}$ such that, for all $(t,\xi) \in \bar{\mathcal{D}}$, we have
    \begin{align*}
        \| \nabla_{\xi_i} b(t,\xi) \| \le B_i,  \quad 
        \| \partial_t b(t,\xi) \| \le T.
    \end{align*}
\end{lem}
\begin{proof} Since $\bar{\mathcal{D}} \subset \mathcal{D}$ is compact, there exist $\delta_1,\ldots,\delta_r\in (0,1)$ such that
$\|\beta_i(t,\xi)\|\le \delta_i$ for all $(t,\xi)\in \bar{\mathcal{D}}$ and~$i \in [r]$. Fix $i\in [r]$. Then  we have 
\[
    \nabla_{\xi_i} b(t,\xi) = -(\nabla_{\xi_i} \beta_r(t,\xi)) \beta_r(t,\xi).
\]
For $j=i,\ldots,r$, we may compute that
\[
\nabla_{\xi_i} \beta_j(t,\xi) = \begin{cases} -\frac{1}{\psi_j(t)} I, & i=j,\\
    -\frac{\nabla_{\xi_i} \beta_{j-1}(t,\xi)}{1-\|\beta_{j-1}(t,\xi)\|^2} & \\
    -\frac{2 \beta_{j-1}(t,\xi) \beta_{j-1}(t,\xi)^\top \nabla_{\xi_i} \beta_{j-1}(t,\xi)}{(1-\|\beta_{j-1}(t,\xi)\|^2)^2}, & i<j \end{cases}
\]
for all $(t,\xi)\in \bar{\mathcal{D}}$. 
We note that $\nabla_{\xi_i} \beta_j(t,\xi)=0$ for~${j<i}$. Since $\beta_0 = 0$, induction over~$j \in [r]$ shows that, if $\|\nabla_{\xi_i} \beta_{j-1}(t,\xi)\|\le C_{i,j-1}$ for some $C_{i,j-1}\ge 0$, $j\in[r]$, and all $(t,\xi)\in \bar{\mathcal{D}}$, then
\[
\|\nabla_{\xi_i} \beta_j(t,\xi)\| \le \begin{cases} \frac{1}{\psi_j(t)} , & i=j,\\
    \frac{C_{i,j-1}}{1-\delta_{j-1}^2} +\frac{2 C_{i,j-1}}{(1-\delta_{j-1}^2)^2}, & i<j,\\
    0, & i>j.
    \end{cases}
\]
Hence, invoking the properties of~$\psi_i$, boundedness of $\|\nabla_{\xi_i} \beta_r\|$ on $\bar{\mathcal{D}}$ follows, which shows boundedness of $\|\nabla_{\xi_i} b\|$. To show boundedness of $\|\partial_t b\|$ on $\bar{\mathcal{D}}$, observe that 
\[
\partial_t b(t,\xi)  =- (\partial_t \beta_r(t,\xi)) \beta_r(t,\xi)
\]
and, for all $i\in [r]$,
\begin{align*}
&\partial_t \beta_i(t,\xi) = 
-\tfrac{y_{\rm ref}^{(r)}(t) \psi_i(t) - y_{\rm ref}^{(i-1)}(t) \dot \psi_i(t) }{\psi_i(t)^2} \\
& + \tfrac{\partial_t \beta_{i-1}(t,\xi)}{1-\|\beta_{i-1}(t,\xi)\|^2} 
+ \tfrac{2 \beta_{i-1}(t,\xi) \beta_{i-1}(t,\xi)^\top \partial_t \beta_{i-1}(t,\xi)}{(1-\|\beta_{i-1}(t,\xi)\|^2)^2}.
\end{align*}
Hence, with a similar induction, invoking $\beta_0=0$ and the properties of $\psi_i$ and $y_{\rm ref}$, boundedness of $\|\partial_t b\|$ follows.
Defining~$B_i, T \ge 0$ appropriately yields the assertion.
\end{proof}

To derive a feedback control law based on the CBF~$b$ defined in~\eqref{eq:CBF}, 
we introduce the following set of candidate controls on the domain~$\mathcal{D}$ given in~\eqref{eq:control_domain}
\begin{equation}\label{eq:SetOfControls}
     \mathcal{U}(t,\xi) = \left. \left\{  \lambda \frac{\nabla_{\xi_r} b(t,\xi)}{b(t,\xi)} \ \right| \ \lambda \in  [\underline{\lambda}, \bar \lambda] \right\}, \ (t,\xi) \in \mathcal{D},
\end{equation}
where the constants $0 < \underline{\lambda} \le \bar \lambda$ are chosen by the user, and the function~$b$ is given in~\eqref{eq:CBF}.
Note that the set $\mathcal{U}(t,\xi)$ is well-defined, since $b(t,\xi) > 0$ for ${(t,\xi) \in \mathcal{D}}$. 
We emphasize that~$\mathcal{U}(t,\xi)$ depends only on~$t$ and~$\xi$. 
In particular, no system data~$f,g,Q$ is used to define the set~$\mathcal{U}(t,\xi)$ of candidate controls.
For a system~\eqref{eq:System_r_InputOoutputForm} with $\xi(t) = \big(y(t), \dot y(t),\ldots,y^{(r-1)}(t)\big)$ this means that evaluating the CBF-based controls requires only output (and derivative) information.
A particular control~$u \in \mathcal{U}(t,\xi)$ resembles a funnel control feedback, cf.~\cite{berger2018funnel}.

\begin{thm}\label[thm]{Thm:ControlsSubsetKCBF}
    Let system~\eqref{eq:System_r_InputOoutputForm} satisfy \Cref{ass:ID_BIBS,Ass:g_positive}.
    Let~$y_{\rm ref}$ and~$\psi_k$, $k \in [r]$, satisfy \Cref{Ass:reference_and_funnel}.
    Then, for any $\bar q > 0$ and any compact~$\bar{\mathcal{D}} \subseteq \mathcal{D}$, there exists a linear function $\alpha \in \mathcal{K}_\infty^e$ such that, for all ${(t,\xi,\eta) \in \R_{\ge 0} \times \R^{rm} \times \R^q}$, we have
    \begin{align*}
        \Big( (t,\xi) \in \bar{\mathcal{D}}  \wedge \|\eta\| \le \bar{q} \Big) 
         \implies  \mathcal{U}(t,\xi) \subseteq K_{\rm CBF}(t,\xi,\eta,\alpha),
    \end{align*}
    where~$K_{\rm CBF}(t,\xi,\eta,\alpha)$ is given in~\eqref{eq:CBFbasedControls}.
\end{thm}
\begin{proof}
By compactness of~$\bar{\mathcal{D}} \subseteq \mathcal{D}$ there exists $\kappa > 0$ such that $\|\xi\| \le \kappa$ for all~$(t,\xi) \in \bar{\mathcal{D}}$.
By the boundedness of~$y_{\rm ref}^{(r)}$, for the functions in~\eqref{eq:System_r_InputOoutputForm} we may define the constants
    \begin{align*}
        \bar f &:= \max_{ (t,z) \in \bar{\mathcal{D}}, \, \|q\| \le \bar q } \| f(z,q) \| + \| y_{\mathrm{ref}}^{(r)}\|_\infty, \\ 
        0 < \underline{g} &:= \min_{\|e\|=1} \min_{ (t,z) \in \bar{\mathcal{D}}, \, \|q\| \le \bar q } \langle e, g(z,q) e \rangle.
    \end{align*}
    Using \Cref{Lem:bounds_on_derivatives} we obtain, with similar calculations as in the proof of \Cref{Prop:bIsOutputCBF},
\begin{align*}
            & [(\partial_t + L_F)b](t,\xi,\eta) + (L_Gb)(t,\xi,\eta) u  \\
            &\ge - \kappa \sum_{i=0}^{r-1} B_i - B_r \bar f - T - \frac{\langle \beta_r(t,\xi), g(\xi,\eta) u \rangle }{ \psi_r(t)} 
        \end{align*}
for all~$(t,\xi) \in \bar{\mathcal{D}}$.
We set $A_0 := \kappa \sum_{i=0}^{r-1} B_i + B_r \bar f + T $ and
$\bar{\psi}_r := \sup_{s \ge 0} \psi_r(s) > 0$.
Inserting some $u \in \mathcal{U}(t,\xi)$ gives for~$(t,\xi) \in \bar{\mathcal{D}}$ and some $\lambda \in  [\underline{\lambda}, \bar \lambda]$ that
\begin{align*}
    & [(\partial_t + L_F)b](t,\xi,\eta) + (L_Gb)(t,\xi,\eta) u  \\
    & \ge - A_0 - \frac{\langle \beta_r(t,\xi), g(\xi,\eta) \lambda \nabla_{\xi_r} b(t,\xi) \rangle }{ \psi_r(t) b(t,\xi)} \\
    &\ge - A_0 + 2 \frac{\underline{\lambda}\, \underline{g}}{\psi_r(t)^2} \frac{\langle \beta_r(t,\xi), \beta_r(t,\xi) \rangle}{1-\|\beta_r(t,\xi)\|^2} \\
    & \ge - A_0 + \frac{2 \underline{\lambda}\, \underline{g}}{\psi_r(t)^2} \, \frac{\|\beta_r(t,\xi)\|^2}{1-\|\beta_r(t,\xi)\|^2} \\
    & = - A_0 - \frac{2\underline{\lambda}\, \underline{g}}{\psi_r(t)^2} 
    + \frac{\underline{\lambda}\, \underline{g}}{\psi_r(t)^2} \frac{1+\|\beta_r(t,\xi)\|^2}{b(t,\xi)}\\
    &\ge - A_0 - \frac{2\underline{\lambda}\, \underline{g}}{\bar{\psi}_r^2} 
    + \frac{\underline{\lambda}\, \underline{g}}{\bar{\psi}_r^2} \frac{1}{b(t,\xi)}.
\end{align*}
Setting $A := A_0 + 2 \underline{\lambda} \underline{g} / \bar \psi_r$ and $a := \underline{\lambda} \underline{g} / \bar \psi_r$, 
the last line in the above estimate is of the form $l(s) = a/s - A$, for ${s=b > 0}$.
For $s_0=2a/A$, we have $l(s_0) = -A/2$ and $l'(s_0) = -A^2/(4a) < 0$. 
Thereby, the tangent of~$l$ at the point~$s_0$ is given by $\tau(s) = l'(s_0) (s-s_0) + l(s_0) = -s A^2/(4a)$. Since $l$ is convex it holds that $l(s) \ge \tau(s)$ for all $s>0$. 
By the fact that ${b(t,\xi)>0}$ for all $(t,\xi) \in \bar{\mathcal{D}}$, we conclude that
$(\partial_t + L_F)b + (L_Gb) u \ge \tau(b)$.
Defining ${\alpha(s)= -\tau(s)}$, we find that
$\alpha\in\mathcal{K}_\infty^e$ and ${u\in K_{\rm CBF}(t,\xi,\eta,\alpha)}$.
Therefore, 
$\mathcal{U}(t,\xi) \subseteq K_{\rm CBF}(t,\xi,\eta,\alpha)$.
\end{proof}

We emphasize that no knowledge about the system data~$f,g$ and $Q$ from system~\eqref{eq:System_r_InputOoutputForm} is used in designing the control set~$\mathcal{U}(t,\xi)$ in~\eqref{eq:SetOfControls}. 
In the following result, we characterize the closed-loop properties when applying inputs~$u\in\mathcal{U}(t,\xi)$ to a system~\eqref{eq:System_r_InputOoutputForm}.
In particular, we show that controls contained in the CBF-based control set $\mathcal{U}(t,\xi)$ render the set~$\mathcal{D} \subset \mathcal{C}$ forward invariant, that is, we show $\|y(t) - y_{\mathrm{ref}}(t)\| < \psi_1(t)$ for~$t\ge0$ in the closed-loop system.

\begin{thm} \label{Thm:FunnelControl}
    Let system~\eqref{eq:System_r_InputOoutputForm} satisfy \Cref{ass:ID_BIBS,Ass:g_positive}, and let $\psi_k$, $k \in [r]$, and $y_{\mathrm{ref}}$ satisfy \Cref{Ass:reference_and_funnel}. 
    Consider a feedback policy $\mu:\R_{\geq0}\times\R^{rm}\to\R^m$ that is measurable in the first argument, continuous in the second, and satisfies ${\mu(t,\xi)\in\mathcal{U}(t,\xi})$ for all $(t,\xi)\in \mathcal{D}$. 
    Then, system~\eqref{eq:System_r_InputOoutputForm} 
    with $(0,\xi(0)) \in \mathcal{D}$ and input ${u(t):=\mu(t,\xi(t))}$ has a globally defined maximal solution, i.e., $(\xi,\eta): [0,\infty) \to \R^{rm} \times \R^q$.
    Moreover,
    \begin{itemize}
        \item[(i)] $u\in L^\infty(\R_{\ge 0};\R^m)$,
        \item[(ii)] there exists~$\varepsilon \in (0,1)$ such that $b(t,\xi(t)) \ge \tfrac{1}{2} (1-\varepsilon^2)$ for all~${t \ge 0}$.
    \end{itemize}
    In particular, for~$\mathcal{C}$ given in~\eqref{eq:SafeSet}, statement~(ii) implies $(t,\xi(t))\in \mathcal{C}$ for all $t \ge 0$, i.e., $\|y(t) - y_{\mathrm{ref}}(t)\| \le \psi_1(t)$.
\end{thm}

To prove \Cref{Thm:FunnelControl}, we state the following auxiliary result, the proof of which is relegated to the appendix.

\begin{lem}\label[lem]{Lem:ExDelta}
    Let $\psi_k$, $k \in [r]$, satisfy \Cref{Ass:reference_and_funnel} and let $\hat{t}\geq 0$ and  $\xi=(\xi_1,\ldots,\xi_r)\in\R^{rm}$ be
    such that  $\|\beta_i(\hat{t},\xi)\| < 1$ for all $i \in [r-1]$.
    Then, for every function ${\zeta \in C^r([\hat{t},\infty);\R^m)}$ with 
    \begin{itemize}
        \item $(t,\zeta(t),\dot{\zeta}(t),\ldots,\zeta^{(r-1)}(t))\in \mathcal{D}_r$ for all $t \ge \hat{t}$,
        \item ${\zeta^{(i-1)}(\hat{t}) =\xi_i}$ for all $i\in[r]$,
        \item $\|\beta_r(t,(\zeta,\dot{\zeta},\ldots,\zeta^{(r-1)})(t))\| \le 1$ for all $t \ge \hat{t}$,
    \end{itemize}
    there exist $\delta_i \in (0,1)$ and~$\mu_i \in \R$ such that
    \begin{itemize}
        \item[(a)] $\|\beta_i(t,(\zeta,\dot{\zeta},\ldots,\zeta^{(r-1)})(t))\| \le \delta_i$ and
        \item[(b)] $\|\ddt \beta_i(t,(\zeta,\dot{\zeta},\ldots,\zeta^{(r-1)})(t))\| \le \mu_i$
    \end{itemize}
    for all $t \ge \hat{t}$ and all $i \in [r-1]$.
\end{lem}

With \Cref{Lem:ExDelta} at hand, we are in the position to provide the following proof of \Cref{Thm:FunnelControl}.
\begin{proof}
Consider the closed-loop system~\eqref{eq:System_r_InputOoutputForm} with input ${u(t)=\mu(t,\xi(t))}$, which is defined on the domain
    \[
        \Omega:= \{ (t,\xi,\eta)\ |\ (t,\xi)\in\mathcal{D} \},
    \]
    that is, $\|\beta_r(t,\xi)\| < 1$ on~$\Omega$.
    Existence of a local maximal solution ${(\xi,\eta): [0,\omega) \to \R^{rm+q}}$, $\omega \in (0,\infty]$, follows from~\cite[Thm.~5]{IlchRyan02b}. Furthermore, the closure of the graph of~$(\xi,\eta)$ is not a compact subset of~$\Omega$. The remainder of the proof proceeds in steps.

\emph{Step 1}: We show that $\xi$ is bounded in $[0,\omega)$. With some abuse of notation, for ${t \in [0,\omega)}$ we write ${\beta_i(t)\coloneqq \beta_i(t,\xi(t))}$ and define the auxiliary functions~$\chi(s) = 1/(1-s)$ and $\kappa_i(t)
    \coloneqq  \chi(\|\beta_i(t)\|^2) \beta_i(t)$, $i\in[r-1]$. 
    For better legibility, we omit the time dependency of functions in the following.  For~${i \in [r-1]}$, each of the auxiliary signals~$\beta_i$ from~\eqref{eq:AuxBeta} satisfies on $[0,\omega)$
    \begin{equation}
    \begin{aligned}  \label{eq:dot_beta_i}
        \dot \beta_i 
        &= \tfrac{-\dot{\psi}_i}{\psi_i^2}(\xi_{i} - y_{\rm ref}^{(i-1)})+\tfrac{\xi_{i} - y_{\rm ref}^{(i-1)}}{\psi_i} + \dot \kappa_{i-1}\\
        &= \tfrac{-\dot \psi_i}{\psi_i} ( \beta_i -\kappa_{i-1} ) 
        + \tfrac{\psi_{i+1}}{\psi_{i}}\rbl \beta_{i+1}  - \kappa_i \rbr
        + \dot \kappa_{i-1} .
    \end{aligned}
    \end{equation}
    \Cref{Lem:ExDelta} implies existence of~$\delta_i \in (0,1)$ such that $\|\beta_i(t)\| \le \delta_i$, and constants~$K_i, \bar K_i \ge 0$, ${i \in [r-1]}$, such that $\|\kappa_i(t)\| \le K_i$ and $\|\ddt \kappa_i(t)\| \le \bar K_i$ for ${t \in [0,\omega)}$. Since $\beta_i = (\xi_i - y_{\rm ref}^{(i-1)})/\psi_i + \kappa_{i-1}$ and we have that $ \beta_i, \psi_i, y_{\rm ref}^{(i-1)}, \kappa_{i-1}$ are bounded, it follows that $\xi_i$ is bounded for $i=1,\ldots,r$, whence boundedness of $\xi$ on $[0,\omega)$.

\emph{Step 2}: We show that $\eta$ is bounded on $[0,\omega)$. Note that, for any $\zeta \in C(\R_{\ge 0};\R^{rm}) \cap L^\infty(\R_{\ge0};\R^{rm})$, it follows from \Cref{ass:ID_BIBS} that the global solution $\eta(\cdot;0,\eta(0),\zeta)$ exists. Then, again by \Cref{ass:ID_BIBS}, there exists~${\bar q>0}$ such that  $\|\eta(\cdot;0,\eta(0),\zeta)\|_\infty \le \bar q$ for any such~$\zeta$ with $\|\zeta\|_\infty\le \|\xi\|_\infty$, showing boundedness of the internal state~$\eta$.

\emph{Step 3}: We show statement~(ii) with~$\varepsilon \in (0,1)$ determined below. Observe that using the constants $\delta_i$ from \Cref{Lem:ExDelta}, the bounds in \Cref{Lem:bounds_on_derivatives} can also be derived for the set
\[
    \bar{\mathcal{D}} = \setdef{(t,\zeta)\in\R_{\ge 0}\times\R^{rm}}{\|\zeta\|\le \|\xi\|_\infty}\subseteq \mathcal{D}.
\]
Based on these and $\bar q>0$ from \textit{Step~2}, the constants~$\underline{g}, \overline{f}, \bar{\psi}_r, A_0$ can be defined as in the proof of \Cref{Thm:ControlsSubsetKCBF}. Now, we set
    \begin{align*}
        \hat \varepsilon &\in (0,1) \ \  \text{s.t.} \ \ \frac{\hat \varepsilon^2}{1- \hat \varepsilon^2} \ge \frac{A_0 \bar{\psi}_r^2}{ 2  \underline{\lambda}\, \underline{g}}, \\
        \varepsilon &:= \max \left\{ \hat \varepsilon, \|\beta_r(0)\|\right\} < 1,
    \end{align*}
and show that $\|\beta_r(t)\| \le \varepsilon $ for all~$t \ge 0$. By construction of~$\varepsilon$,
    we have $\|\beta_r(0)\| \le \varepsilon $.
    Seeking a contradiction, suppose the existence of~$t^* \in (0,\omega)$ such that ${\epsilon < \|\beta_r(t^*)\| < 1}$.
    Continuity implies the existence of~$t_* := \max\{ t \in [0,t^*) \, | \, \|\beta_r(t)\| = \varepsilon  \}$.
    Therefore, we have $\|\beta_r(t)\|\ge \varepsilon$ and hence $\|\beta_r(t)\|^2/(1 - \|\beta_r(t)\|^2) \ge \varepsilon^2/(1- \varepsilon^2)$ for all $t \in [t_*, t^*]$.
    On $[t_*,t^*]$, we calculate
    \begin{equation*}
        \begin{aligned}
            \tfrac{1}{2}\tfrac{\text{d}}{\text{d} t}  \|\beta_r\|^2  & = \langle \beta_r, \dot \beta_r \rangle 
            =\langle \beta_r, \partial_t \beta_r + (\nabla_\xi \beta_r)^\top \dot \xi\rangle\\
            &= \Big\langle \beta_r, \partial_t \beta_r + \sum_{i=1}^{r-1} (\nabla_{\xi_i} \beta_r)^\top \xi_{i+1} \\ 
            &\quad + (\nabla_{\xi_r} \beta_r)^\top \big(f(\xi,\eta) + g(\xi,\eta) u\big) \Big\rangle\\   
            &= -\big(\partial_t b + L_F b + (L_G b)u \big),
            \end{aligned}
    \end{equation*}
    where the last equality follows with the same calculations as in the proof of \Cref{Prop:bIsOutputCBF}. Then, with the same estimates  as in the proof of \Cref{Thm:ControlsSubsetKCBF} we find that
    \begin{equation*}
        \begin{aligned}
            \tfrac{1}{2}\tfrac{\text{d}}{\text{d} t}  \|\beta_r\|^2  & \le A_0 - \frac{2\underline{\lambda}\, \underline{g}}{\bar{\psi}_r^2}  \frac{\|\beta_r\|^2}{1-\|\beta_r\|^2} \\
             &\le A_0 - \frac{2\underline{\lambda}\, \underline{g}}{\bar{\psi}_r^2}
             \frac{\varepsilon^2}{1-\varepsilon^2} \le 0.
        \end{aligned}
    \end{equation*}
    Then, the contradiction $\varepsilon  < \| \beta_r(t^*)\| \le \|\beta_r(t_*)\| = \varepsilon $ arises upon integration.

    \emph{Step 4}: It follows from \textit{Step~3} that $\omega = \infty$, as otherwise the graph of the solution~$(\xi,\eta)$ evolves within a compact subset of~$\Omega$. This, in combination with \textit{Step~3}, implies that ${b(t,\xi(t)) \ge \tfrac{1}{2}(1-\varepsilon^2) > 0}$ for all $t \ge 0$.
    Moreover, we may estimate the maximal control input by 
    \begin{align*}
        \|u(t)\| &= \lambda \frac{\nabla_{\xi_r} b(t,\xi(t))}{b(t,\xi(t))} = -\frac{2\lambda}{\psi_r(t)} \frac{\beta_r(t)}{1-\|\beta_r(t)\|^2} \\
        & \le \frac{2 {\overline{\lambda}}}{\inf_{s\ge 0} \psi_r(s)} \frac{\varepsilon}{1-\varepsilon^2}
    \end{align*} 
    for all $t \ge 0$, i.e, $u \in L^\infty(\R_{\ge 0};\R^m)$.
    \end{proof}

We highlight that, in virtue of \Cref{Lem:ExDelta}, it suffices that the feedback is chosen such that~$\|\beta_r(t,\xi(t))\| \le 1$ for all~$t \ge 0$ in the closed-loop system in order to ensure that the variables~$\beta_k$ are well-defined on $\mathcal{D}_{k-1}$, $k \in [r]$, and to achieve the control objective, namely forward invariance of the safe set~$\mathcal{C}$ given in~\eqref{eq:SafeSet}.
\Cref{Thm:FunnelControl} summarizes this observation for the closed-loop behavior of a system~\eqref{eq:System_r_InputOoutputForm} with input $u \in \mathcal{U}(t,\xi)$.

\subsection{Discussion - model-free CBF}
The proposed high-order CBF~\eqref{eq:CBF} is constructed in a \emph{model-free} way by using insights from funnel control. 
In particular, the CBF in~\eqref{eq:CBF} and the set of safe control inputs~$\mathcal{U}(t,\xi)$ in~\eqref{eq:SetOfControls} do not rely on knowledge of the
model equations~\eqref{eq:System_r_InputOoutputForm}, but only on the  
structural \Cref{ass:ID_BIBS,Ass:g_positive}. 
In contrast, standard CBFs~\cite{ames2016control} rely on (accurate) model knowledge to define a suitable CBF and characterize the set of safe inputs $K_{\rm CBF}(t,x)$ in~\eqref{eq:CBFbasedControls}.
We note that ``model-free'' CBFs are also proposed in~\cite{molnar2021model} and~\cite[Sec.~5]{cohen2024safety}, which are particularly suitable for fully actuated mechanical systems of the form
\begin{align} \label{eq:FARS}
 \!\! D(y(t))\ddot{y}(t)=C(y(t),\dot{y}(t))\dot{y}(t)(t)+G(y(t))+Bu(t)
\end{align}
with position $y(t)\in\R^m$, velocity $\dot{y}(t)\in\R^m$ and control input $u(t)\in\R^m$ at time $t\ge 0$. Furthermore, ${D:\R^m\to\R^{m\times m}}$ is the positive definite position-dependent mass matrix, $C:\R^m\times\R^m\to\R^{m\times m}$ the centrifugal and Coriolis force matrix, $G:\R^m\to\R^m$ the gravity vector and $B\in\R^{m\times m}$ the positive definite input distribution matrix. 
To satisfy a desired constraint $h(y)\geq 0$, the methods from~\cite{molnar2021model} and~\cite{cohen2024safety} choose a desired velocity $v_{\mathrm{d}}$, e.g., based on the gradient of the constraint $h$, and then use a velocity controller to ensure an exponential decrease in the following 
Lyapunov-like function for the velocity error 
$V(y,\dot{y},v_{\mathrm{d}}(t))={(\dot{y}-v_{\mathrm{d}}(t))^\top D(y)(\dot{y}-v_{\mathrm{d}}(t))}$.
If the velocity controller is sufficiently fast, then
\begin{align}
\label{eq:CBF_FARS}
{b(t,y,\dot y) :=} h(y)-\frac{1}{2\mu}V(y,\dot{y},v_{\mathrm{d}}(t))\geq 0
\end{align}
defines a CBF for the full dynamics with a tunable weight~${\mu>0}$. 
Clearly, the method we propose is also 
applicable to this problem class (no internal dynamics, relative degree~$r=2$, and positive definite input distribution matrix). In fact, for $e=y-y_{\rm ref}$, $v_{\mathrm{d}} = \dot y_{\mathrm{ref}} $ and constant~$\psi_1(t) = \psi_2(t) = 1$ (to ease notation), we could interpret the term  
\[
\|\beta_2(t,y(t),\dot y(t)) \| = \| \dot e(t) + e(t)/(1-\|e(t)\|^2)\|
\] 
in the CBF~\eqref{eq:CBF} as a term tracking a desired velocity. 
A key difference is that we can flexibly deal with time-varying reference signals and additionally optimize over the applied input using the scaling parameter~$\lambda$, while~\cite{molnar2021model} can only indirectly change the velocity reference and requires suitable tuning. 
More importantly, the CBF~\eqref{eq:CBF_FARS} cannot be evaluated without model knowledge and the method is only applicable to a particular class of mechanical systems. 
In contrast, the proposed CBF~\eqref{eq:CBF} depends only on the output (including derivatives) and the funnel functions~$\psi_i$, which can be evaluated and utilized without  knowing the dynamics or a Lyapunov function. 
Furthermore, our method is applicable to higher-order systems with non-trivial internal dynamics. 

\section{Numerical example} \label{Sec:numerical_example}
In the following, we demonstrate the applicability and benefits of the proposed method by means of the example of a seven degree of freedom robotic manipulator. The code to reproduce the results, including the definition of all parameters is available online:\\
{\url{https://github.com/KohlerJohannes/Funnel_CBF}}
\begin{figure}
    \centering
    \includegraphics[width=0.45\linewidth]{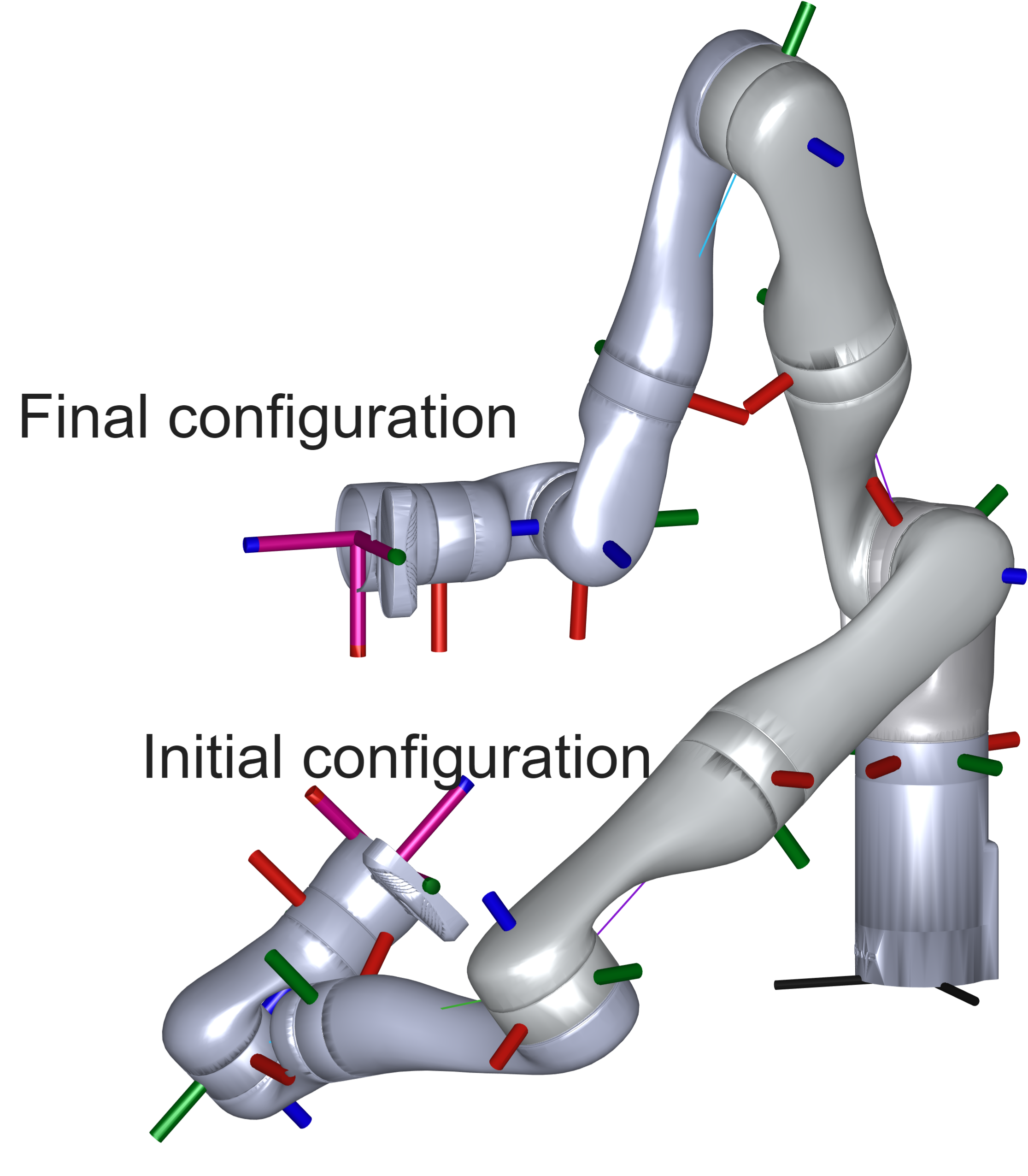}
    \caption{Initial and target configurations of the Kinova Gen3 manipulator used in simulations.}
    \label{fig:robot}
\end{figure}
 
\textit{Setup:} 
We consider a Kinova Gen3 robotic manipulator (cf. Fig.~\ref{fig:robot}) with the dynamics supplied by the Robotics Systems Toolbox on Matlab~\cite{corke2017robot}. 
The robot dynamics are characterized by a nonlinear ODE of the form~\eqref{eq:System_r_InputOoutputForm}. The system has $m=7$ outputs and control inputs, corresponding to the measured configuration angle and the applied torque. The system has a well-defined relative degree of $r=2$. 
\Cref{ass:ID_BIBS} is trivially satisfied as the system lacks internal dynamics, i.e., $q=0$. 
\Cref{Ass:g_positive} holds due to the positive definiteness of the mass matrix of the manipulator dynamics. 
We do not provide the exact system dynamics here, which is intended to illustrate that we do not need them in order to define the CBF~\eqref{eq:CBF} and implement a feedback based on it; the aforementioned structural knowledge is sufficient. 

\textit{Control goal:}
We seek tracking of a smooth reference trajectory $y_{\mathrm{ref}}$ moving the robot from a starting configuration to a desired target configuration, which are visualized in~\Cref{fig:robot}. 
The primary control objective is to stay within the error bound~\eqref{eq:ErrorGuarantee}. 
The funnel size $\psi$ should account for the distance between the planned trajectory and obstacles, see, e.g., \cite{wullt2025probabilistic} for efficient computations of such a safe radius in complex robot geometries.
For the following simulations, we set $\psi_1(t)=10^{-2}$ and $\psi_2(t)=1$ in~\eqref{eq:AuxBeta}.  

\textit{Secondary goal and CBF Implementation:}
To highlight the flexibility of the CBF-based approach, we demonstrate how a secondary objective can be incorporated without compromising safety. 
Specifically, we define a desired torque reference $u_{\mathrm{ref}}(t)$ consisting of a stationary torque with superimposed oscillations. 
These oscillations make the torque visually distinguishable in the simulation and can aid in identifying uncertain system dynamics. 
We implement the controller $u(t)=\lambda(t) \frac{\nabla_{\xi_r} b(t,\xi(t))}{b(t,\xi(t))}$ (cf.~\eqref{eq:SetOfControls}). We determine $\lambda(t)\in[\underline{\lambda},\overline{\lambda}]=[10^{-2},10^2]$ by 
minimizing $\|u(t)-u_{\mathrm{ref}}(t)\|^2$ every $5~[ms]$, yielding a piece-wise constant $\lambda(t)$. This computation of $\lambda(t)$ admits a closed-form expression, which only depends on the input reference $u_{\mathrm{ref}}(t)$ and the measured output errors with derivatives.

\begin{figure}[t!]
    \centering
    \includegraphics[width=0.45\textwidth]{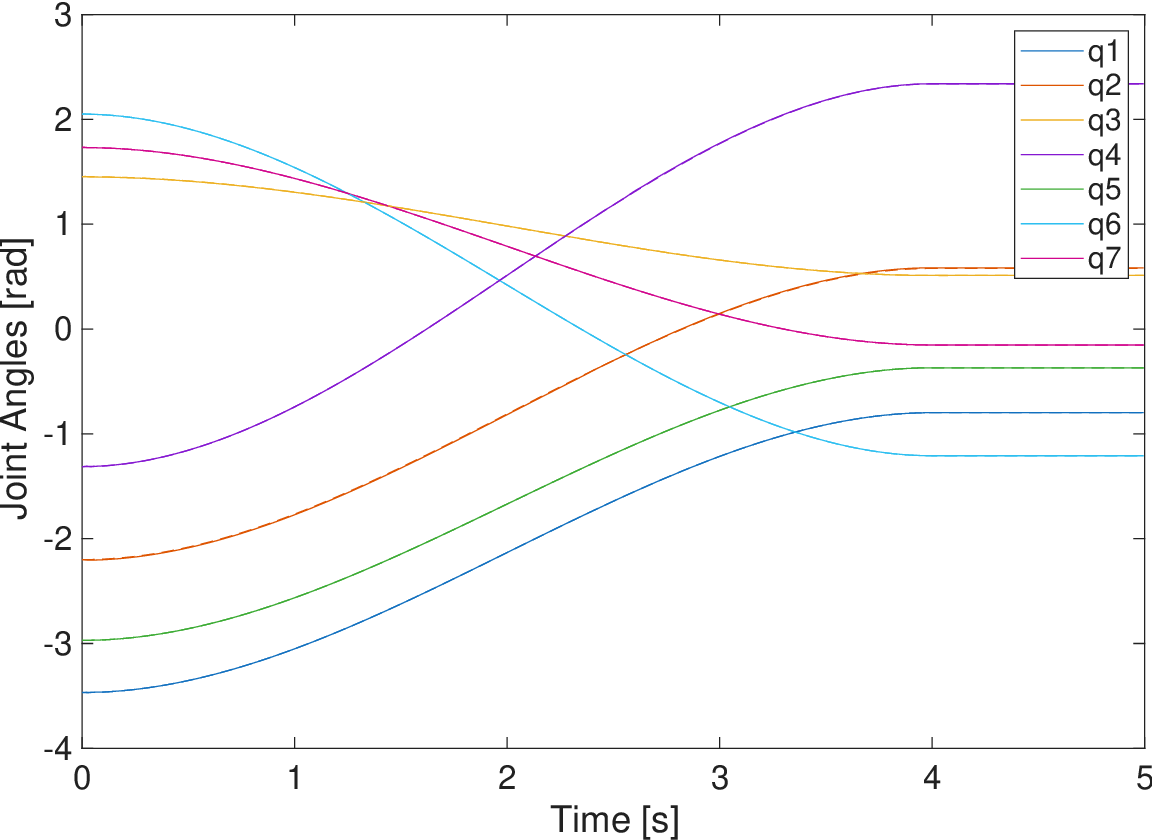}
    \includegraphics[width=0.45\textwidth]{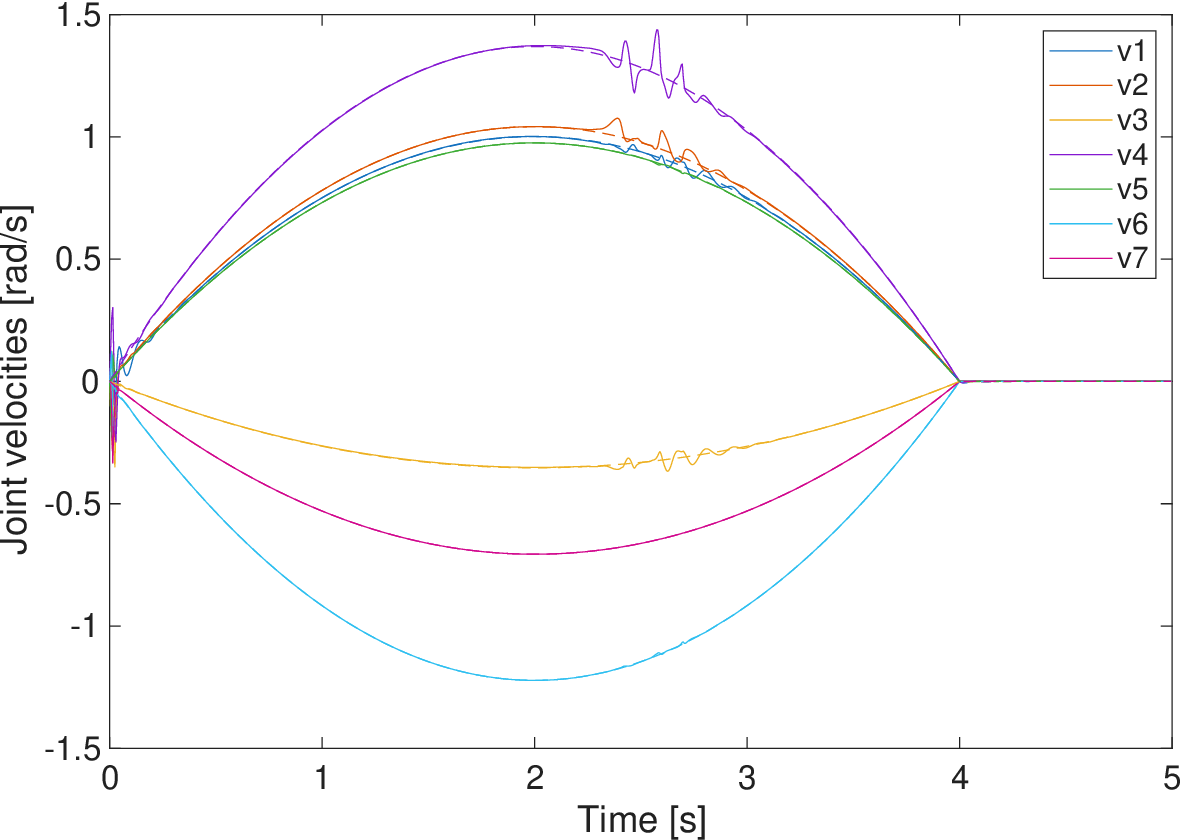}
    \includegraphics[width=0.45\textwidth]{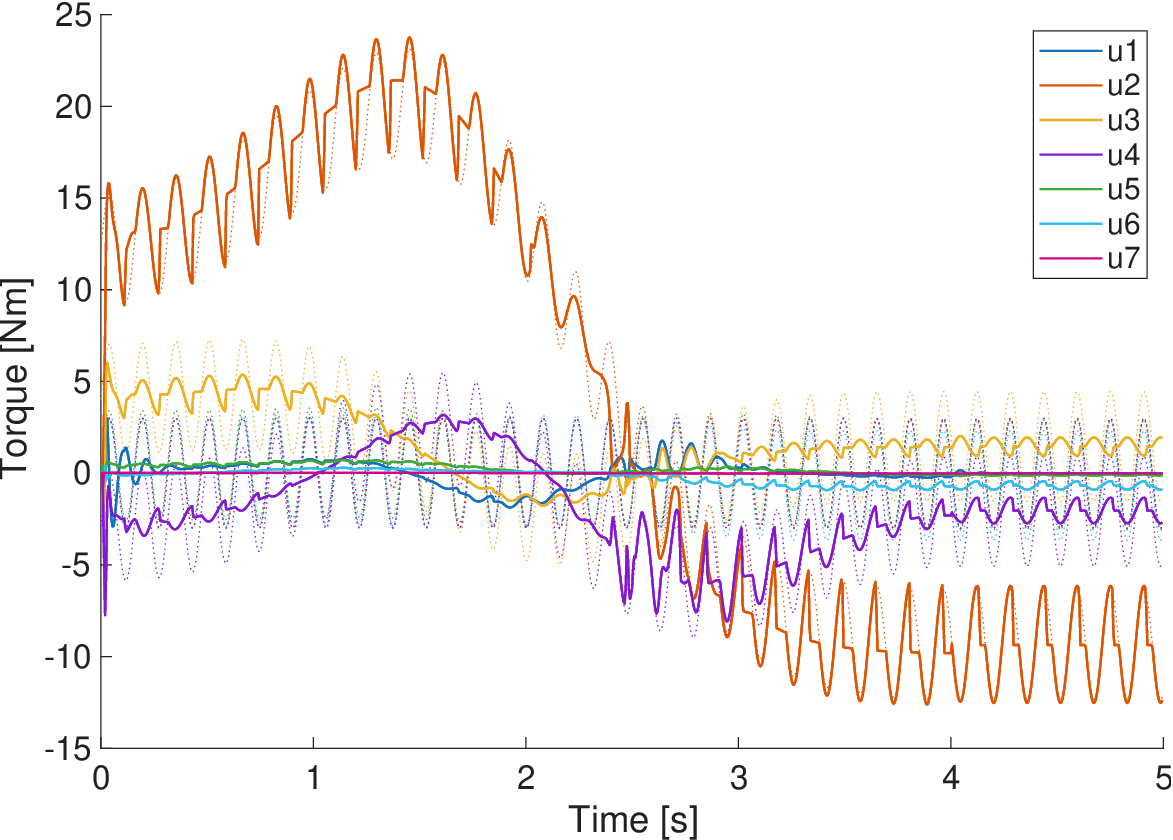}
    \caption{Simulation results of the proposed controller.
    Top: Configuration $y$. Middle: Configuration velocity $\dot y$. Bottom: Applied torque $u$. The trajectories are shown as solid lines and the references are dotted.}
    \label{fig:closed_loop}
\end{figure}

\textit{Results:}
Simulation results are shown in~\Cref{fig:closed_loop}. 
The proposed controller keeps the manipulator output within the prescribed performance funnel throughout the trajectory, demonstrating that the primary safety objective is satisfied. The applied torque closely follows the reference $u_{\mathrm{ref}}(t)$ whenever possible, particularly when the system is far from the funnel boundary. This behavior is especially visible in the oscillations of the torque.

\section{Conclusion}
This article addressed the problem of using control barrier functions (CBFs) in a model-free fashion for nonlinear systems of arbitrary  relative degree. 
In particular, our goal is to enforce that a system output remains within a prescribed funnel around a desired reference while providing flexibility to optimize inputs for arbitrary user-chosen criteria.
We provided a simple example showing that the existing notion of high-order CBFs is not suitable to address this problem. 
Instead, we utilized techniques from funnel control to characterize a subset of the controls satisfying a CBF condition without requiring a dynamic model or state measurement. 
A numerical example with a seven degree of freedom robotic manipulator shows that the proposed approach is applicable to systems with higher relative degree and unknown dynamics, enforcing the desired bounds.

\bibliographystyle{elsarticle-harv}
\bibliography{references}

\section*{Appendix}
\begin{proof}[Proof of \Cref{Lem:ExDelta}.]
    We adapt 
    the proof of~\cite[Lem.~2.2]{lanza2024sampled} to the current setting.
    For $i \in [r-1]$, we define ${\sigma_i:=\inf_{s\geq \hat t}\|\tfrac{\psi_{i+1}(s)}{\psi_{i}(s)}\|}$.
    Setting~${\delta_0:=0}$, $\bar{\kappa}_0:=0$ and utilizing the 
    the short hand notation $\chi(s)=\tfrac{1}{1-s}$, we define successively for $i \in [r-1]$
    \begin{align*}
    &\hat \delta_i \!\in (0,1)  \text{ s.t. } \\
        &\chi(\hat \delta_i^2) \hat \delta_i\geq  \!\!\tfrac{1}{\sigma_i}\!\rbl\SNorm{\tfrac{\dot{\psi}_i}{\psi_i}}\!\!\!\!\! ( 1 + \chi(\delta_{i-1}^2) \delta_{i-1})+ \SNorm{\tfrac{\psi_{i+1}}{\psi_{i}}}\!\!\!\!\!   \!+\!  \bar \kappa_{i-1}\!\rbr\!\!, \\
        &\delta_i \!\coloneqq\!  \max \{ \| \beta_i(\hat{t},\xi)\|,  \hat \delta_i\} < 1,\\
        &\mu_i \! \coloneqq\!  \SNorm{\tfrac{\dot{\psi}_i}{\psi_i}} \!\!\!\!\!\!( 1  \!+\!  \chi(\delta_{i-1}^2) \delta_{i-1} ) \! + \!\!\SNorm{\tfrac{\psi_{i+1}}{\psi_{i}}}\!\!\!\!\!   \rbl 1 \!+\!  \chi(\delta_i^2) \delta_i\rbr \!+\! \bar \kappa_{i-1},\\
        &\bar \kappa_i \! \coloneqq\!  2 \chi(\delta_i^2)^2 \delta_i^2 \mu_i + \chi(\delta_i^2) \mu_i.
    \end{align*}
    Let $\zeta \in C^r([\hat{t},\infty);\R^m)$ be a function with
    ${\zeta^{(i-1)}(\hat{t}) =\xi_i}$ for all $i\in[r]$. With some abuse of notation, we write ${\beta_i(t)\coloneqq \beta_i(t,(\zeta,\dot{\zeta},\ldots,\zeta^{(r-1)})(t))}$ and define the auxiliary functions $\kappa_i(t)
    \coloneqq  \chi(\|\beta_i(t)\|^2) \beta_i(t)$. 
    By assumption, $\|\beta_r(t)\|\le 1$ and ${(t,\zeta(t),\dot{\zeta}(t),\ldots,\zeta^{(r-1)}(t))\in \mathcal{D}_r}$ for all $t \ge \hat{t}$, thus $\|\beta_i(t)\| <1$ for all $i\in[r-1]$.
    We show that~(a) holds, that is $\|\beta_i(t)\|^2 \le \delta_i$ for all $t\geq\hat{t}$ and all $i \in [r-1]$. Seeking a contradiction, assume that there exist $j \in [r-1]$ and ${t^\star > \hat{t}}$ such that $\|\beta_j(t^\star)\|^2 > \delta_j$. W.l.o.g. we assume that this is the smallest index~$j$ with this property, that is, in particular, $\|\beta_{j-1}(t)\|\le \delta_{j-1}$ for all $t\in [\hat t,t^\star]$ and hence $\|\kappa_{j-1}(t)\|\le \chi(\delta_{j-1}^2) \delta_{j-1}$ by monotonicity of $\chi(\cdot)$.
    For better legibility, we omit the time dependency of functions in the following.
    For~${i \in [r-1]}$, each of the auxiliary signals~$\beta_i$ from~\eqref{eq:AuxBeta} satisfies~\eqref{eq:dot_beta_i} with $\xi_i = \zeta^{(i)}$
    on $[\hat t,t^\star]$. Additionally, observe that
        \begin{align*}
        \dot \kappa_i &= 2 \chi(\| \beta_i\|^2)^2 \langle \beta_i, \dot \beta_i \rangle \beta_i + \chi(\| \beta_i\|^2) \dot \beta_i.
    \end{align*}
    Then, by minimality of~$j$, we find that
    \begin{align*}
        \|\dot \kappa_{j-1}\| &\le 2 \chi(\delta_{j-1}^2)^2 \delta_{j-1}^2 \|\dot \beta_{j-1}\| + \chi(\delta_{j-1}^2) \|\dot \beta_{j-1}\|,\\
       \|\dot \beta_{j-1}\| &\le \SNorm{\tfrac{\dot{\psi}_{j-1}}{\psi_{j-1}}} ( 1 +  \chi(\delta_{j-2}^2) \delta_{j-2} ) \\
       &\quad + \SNorm{\tfrac{\psi_{j}}{\psi_{j-1}}}\rbl 1 +  \chi(\delta_{j-1}^2) \delta_{j-1}\rbr + \|\dot\kappa_{j-2}\|.
    \end{align*}
    Induction over~$j$ shows that $\|\dot \beta_{j-1}\|\le \mu_{j-1}$ and $\|\dot\kappa_{j-1}\|\le \bar \kappa_{j-1}$.
Invoking the assumptions, we have that 
\begin{align*}
    \| \beta_j(\hat t) \| & = \|\beta_j(\hat t,(\zeta,\dot{\zeta},\ldots,\zeta^{(r-1)})(\hat t))\| \le 
    \delta_j
\end{align*}
by definition of $\delta_i$, and together with continuity of the involved functions, 
${t_\star \coloneqq  \max \{ t \in [\hat t,t^\star) \ | \ \| \beta_j(t) \| = \delta_j} \}$ is well-defined. Therefore, $\|\beta_j\|\ge \delta_j$ on $[t_\star,t^\star]$, thus
    \[
        \langle \beta_j,\kappa_j\rangle = \chi(\|\beta_j\|^2) \|\beta_j\|^2 \ge  \|\beta_j\|\, \chi(\delta_j^2) \delta_j
    \]
on $[t_\star,t^\star]$. Utilizing~\eqref{eq:dot_beta_i}, we calculate on $[t_\star,t^\star]$
    \begin{align*}
        & \tfrac{1}{2} \dd{t}\Norm{\beta_j}^2 \\
         &=\al \beta_j, \tfrac{-\dot \psi_j}{\psi_j} ( \beta_j \!-\! \kappa_{j-1} ) \!+\! \tfrac{\psi_{j+1}}{\psi_{j}}\rbl\beta_{j+1}  \!-\! {\kappa_j}\rbr \!+\! \dot \kappa_{j-1}  \ar\\
         & \le  \| \beta_j \|  \Big( \SNorm{\tfrac{\dot \psi_j}{\psi_j}} \! ( 1 + \chi(\delta_{j-1}^2) \delta_{j-1}) 
         +  \SNorm{\tfrac{\psi_{j+1}}{\psi_{j}} } \!+  \bar \kappa_{j-1} \Big) \\
         & \qquad - \|\beta_j\|\, \sigma_j\chi(\delta_j^2) \delta_j
        \le  0,
    \end{align*}
    where the last inequality follows from the definition of $\delta_i$ and choice of $\hat \delta_i$. 
    Hence, the contradiction 
    ${\delta_j < \| \beta_j(t^\star)\|^2 \leq \| \beta_j(t_\star)\|^2 = \delta_j}$
    arises after integration. 
    Thus, $\| \beta_i(t) \| \le \delta_i$
    and $\|\dot \beta_i(t)\| \le \mu_i$ for all $i \in [r-1]$ and all~$t \geq 0$, showing~(a) and~(b).
    \end{proof}
\end{document}